\documentclass[12pt,reqno]{amsart}

\usepackage[margin=3.5cm]{geometry}
\usepackage[T1]{fontenc}                                   %imposta la codifica dei font
\usepackage{amsfonts}
\usepackage[utf8]{inputenc}                               %lettere accentate da tastiera
\usepackage{comment}                                       %pacchetto commenti
\usepackage{mparhack}    
\usepackage{eucal}
\usepackage{amsmath,amssymb,amsthm,mathrsfs,eucal}                      %matematica
\usepackage{booktabs}                                                                          %tabelle
\usepackage{graphicx,subfig}                                                                %figure
\usepackage{wrapfig}                                                                            %figure e tabelle nel testo
\usepackage[colorlinks=true, allcolors=blue]{hyperref} 
\hypersetup{urlcolor=blue, citecolor=red, linkcolor=blue}                %per avere i segnalibri
\usepackage{bm}                                                                                   %grassetto
\usepackage{cite}

\newtheorem{thm}{Theorem}[section]
\newtheorem{prop}{Proposition}[section]
\newtheorem{lem}{Lemma}[section]

\newtheorem{rem}{Remark}[section]

\newcommand{\R}{\mathbb{R}}

\newcommand{\supp}{{\rm supp}}

\def\ep{\varepsilon}

\def\Fish{\mathcal E_{\rm kin}}

\def\mycolor(#1){%
   \ifnum#1=2
     green%
   \else
   \ifnum#1=3
     blue%
   \else
   \ifnum#1=4
     red%
   \else
   \ifnum#1=5
     violet%
   \else
   \ifnum#1=6
     yellow%
   \else
   \ifnum#1=7
     purple%
   \else
     black%
   \fi
   \fi
   \fi
   \fi
   \fi
   \fi
   }
   
\usepackage{color}
\usepackage{pstricks}
\usepackage{ifthen}
\newrgbcolor{gerocolor}{.1 .1 .7}
\newrgbcolor{paolacolor}{.7 .1 .1}
\newrgbcolor{augucolor}{.5 .1 .7} %112	51	173
\usepackage{csquotes}

\newboolean{DEBUG}
\setboolean{DEBUG}{true}

\ifthenelse {\boolean{DEBUG}}
{}

\ifthenelse {\boolean{DEBUG}}
{}

\begin{document}
\author{S. Di Marino, A. Gerolin, L. Nenna}
\address{Simone Di Marino, Universit\`a degli studi di Genova (DIMA), MaLGA, Genova, Italy}
\email{simone.dimarino@unige.it}
\address{Augusto Gerolin, Department of Mathematics and Statistics $\&$ Department of Chemistry and Biomolecular Sciences, University of Ottawa, Ottawa, Canada}
\email{agerolin@uottawa.ca}
\address{Luca Nenna, Universit\'e Paris-Saclay, CNRS, Laboratoire de math\'ematiques d'Orsay, 91405, Orsay, France.} 
\email{luca.nenna@universite-paris-saclay.fr}
\title{Universal diagonal estimates for minimizers of the Levy-Lieb functional}
\date{}
\begin{abstract}
\noindent
Given a  wave-function minimizing the Levy-Lieb functional, the intent of this short note is to give an estimate of the probability of the particles being in positions $(x_1, \ldots, x_N)$ in the $\delta$-close regime $D_{\delta}= \cup_{i \neq j} \{|x_i - x_j| \leq \delta\}$.
\end{abstract}
\maketitle
\tableofcontents
\section{Introduction}

Density Functional Theory attempts to describe all the relevant information about a many-body quantum system at ground state in terms of  the one electron density $\rho$. 
Following the Levy and Lieb's approach \cite{Levy-79,Lieb-83b} the ground state energy can be rephrased as the following variational principle involving only the electron density 

$$ \mathcal E_0[v] = \inf_{\substack{\sqrt{\rho}\in H^1(\R^3)\\ \rho(\R^3)=N \\ \int_{\R^3} v (x) \, d \rho<+\infty }} \left\{  \mathcal  F_{LL,\ep}[\rho] +  \int_{\R^3} v (x) \, d \rho \;\right\},$$
where $v$ is an external potential and the Levy-Lieb functional $\mathcal F_{LL,\ep}$ is defined as

\begin{equation}
\label{Fll}
\mathcal F_{LL,\ep}[\rho]:=\min_{\substack{ \psi \in\mathcal H^N \\ \psi \mapsto \rho}}\left\{ \int_{\R^{3N}}\ep|\nabla\psi|^2 (x)+v_{ee}(x) | \psi|^2 (x)  \, d x  \right\},
\end{equation}
where $\mathcal H^N=\bigwedge_{i=1}^N H^1 ( \R^{3}; \mathbb{C})$ is the fermionic space, $v_{ee} (x_1, \ldots, x_N)= \sum_{i <j} \frac 1{|x_i-x_j|}$ is the Coulomb interaction potential between the $N$ electrons and $\psi\mapsto \rho$ means that the one-body density of $\psi$ is $\rho$, that is $\rho=N\int_{\mathbb{R}^{3(N-1)}}|\psi|^2$. 

The Levy-Lieb functional is indeed the lowest possible (kinetic plus interaction) energy of a quantum system having the prescribed density $\rho$. This universal functional is the central object of Density Functional Theory, since knowing it would allow one to compute the ground state energy of a system with any external potential $v$. For a complete review on it we refer the reader to \cite{lewin2019universal}. 

Although the electrons are \textit{fermions}, in this article we will treat only the case of \textit{bosonic} wave-functions, i.e. consider in the constraint search in \eqref{Fll} wave-functions $\psi$ that are symmetric rather than anti-symmetric. Notice that the ground-state energy of \textit{bosonic} systems are generally higher than \textit{fermionic} ones. In our analysis, however, the \textit{bosonic} case is not very restrictive since we are looking at the regime $\ep$ small.

Our approach interprets the Levy-Lieb functional as a (Fisher-information regularized) multi-marginal optimal transport problem.
\vspace{2mm}

\noindent
\textit{Connection with Optimal Transportation Theory:} It has been recently shown that the limit functional as $\ep \to 0$ corresponds to a multi-marginal optimal transport problem \cite{BinPas-17,CotFriKlu-13,CotFriKlu-18,Lewin-18} (see also the seminal works in the physics and chemistry literature \cite{ButPasGor-12,Seidl-99,SeiPerLev-99,SeiGorSav-07,SeiMarGerNenGieGor-17}): rather than wave-functions, one has now enlarged the constrained search in \eqref{Fll} to minimize among probability measures on $\R^{3N}$ having $\rho$ as marginal, that is
\begin{equation}
\label{eq:MK}
\mathcal F_0[\rho]:=\inf_{\mathbb P\in\Pi_N(\rho)}\left\{ \int_{\R^{3N}}v_{ee}(x_1,\ldots,x_N)\, d\mathbb P(x_1,\ldots,x_N)\right\},
\end{equation}
where $\Pi_N(\rho)$ denotes the set of probability measures on $\R^{3N}$ having $\rho/N$ as marginals. 

The multi-marginal optimal transport with Coulomb cost \eqref{eq:MK} has garnered attention in the mathematics, physics and chemistry communities and the literature on the subject is growing considerably. Recent developments include results on the existence and non-existence of Monge-type solutions minimizing \eqref{eq:MK} (e.g., \cite{ColMar-15,ColPasMar-15,CotFriKlu-13,Friesecke-19,ButPasGor-12,ColStr-M3AS-15,DiMGerNenSeiGor-xxx-15,BinDepKau-arxiv-20}), structural properties of Kantorovich potentials (e.g., \cite{ColMarStr-19, MarGerNen-17,GerKauRaj-ESAIMCOCV-19,ButChaPas-17}), grand-canonical optimal transport \cite{MarLewNen-19_ppt}, efficient computational algorithms (e.g., \cite{BenCarNen-SMCISE-16,FriSchVoe-21, CoyEhrLom-19, MarGer-19_ppt, KhoLinLinLex-19_ppt}) and the design of new density functionals (e.g., \cite{GerGroGor-JCTC-20, CheFri-MMS-15, MirUmrMorGor-JCP-14, LanMarGerLeeGor-PCCP-16}). The first order expansion around the limit $\ep\to0$ of the Levy-Lieb functional was obtained in \cite{ColMarStr-21_ppt}.  

We refer to the surveys (and references therein) \cite{MarGerNen-17,FriGerGor-arxiv-22} for a self-contained presentation on multi-marginal optimal transport approach in Density Functional Theory as well as the review article \cite{VucGer-2022} for a the recent developments from a chemistry standpoint.\vspace{2mm}

\noindent
\textbf{Main result of this paper:} In \cite{Pascale-15,ButChaPas-17,ColMarStr-19} it is shown that the support $\supp(\mathbb P^*)$ of a solution $\mathbb P^*$ of the limiting problem \eqref{eq:MK} is uniformly bounded away of the diagonal, i.e. one has always $| x_i-x_j| \geq \delta >0$ for any $x_i,x_j \in \supp(\mathbb P^*)$. In other words, the electrons are always at a certain distance away from each other, which is the expected behaviour since we are in a \emph{classical} framework.

In the sequel we will denote with $D_{\delta}$ the enlarged diagonal
$$D_{\delta} = \{ (x_1, \ldots, x_N) \in \R^{3N} \; :  \; \exists~ i \neq j \text{ s.t. } | x_i - x_j | \leq \delta \}. $$
In particular the result in \cite{ButChaPas-17,ColMarStr-19} can be rephrased saying that the solution to the multi-marginal optimal transport problem is concentrated on the complement of $D_{\delta}$ for some $\delta$.  An important feature of the results is that $\delta$ depends only on concentration properties of $\rho$. In fact defining 
$$\kappa(\rho, r):= \sup_{x \in \R^3} \rho ( B(x,r))/N, $$
the authors in \cite{ColMarStr-19} prove that if $\kappa(\rho, \beta) < \frac 1{2(N-1)}$ then one can choose $\delta = \frac \beta{2N}$. Our main result is to extend this property also for  $\ep>0$ small. In particular we do not expect to have $\psi_{\ep} =0$ on $D_{\delta}$ but we show that the probability of having the electrons in position $x \in D_{\delta}$ is very small \eqref{eqn:polysmall}.
\begin{thm}[Exponential off-diagonal localization for Coulomb]
\label{thm:mainThm}
Let $\psi_\ep$ be a minimizer for \eqref{Fll} where $v_{ee} (x_1, \ldots, x_N)= \sum_{i <j} \frac 1{|x_i-x_j|}$.
Let us consider $\beta$ such that $\kappa(\rho,\beta)\leq \frac{1}{4(N-1)}$ then, let $\alpha\leq \frac \beta{32N}$, and suppose $\ep N^2 \ll \alpha/16$. Then, for $\mathbb{P}_{\ep}(x) = | \psi_{\ep}|^2(x)$ we have
\begin{equation}\label{eqn:polysmall} 
\boxed{\int_{ D_{ \alpha}} \mathbb{P}_\ep(x)\, dx \leq e^{ - \frac 1{24}\sqrt{\frac{ \alpha}{\ep}}}.}
 \end{equation}
\end{thm}
In the proof we actually work with a general repulsive pairwise potential $v_{ee}$, which satisfy the hypothesis \eqref{eqn:Vee2}, stated in the next section. The result in general is the following one:
%Moreover, this asymptotics are uniform in the following sense:   there exist $C$ and $\gamma$, depending only on $\beta$, and not on $\ep$ nor (explicitely) on $\rho$, such that  
%\begin{equation}  \int_{ D_{ \alpha}} |\psi_\ep|^2\, dx \leq \dfrac{C}{\ep^\gamma} 2^{-\ln_2\left(\sqrt{\ep}\right)^2}. \end{equation}

\begin{thm}[Exponential off-diagonal localization for general interaction cost]
\label{thm:mainThm2}
Let $\psi_\ep$ be a minimizer for \eqref{Fll} where $v_{ee}$ satisfies \eqref{eqn:Vee2}.
Let us consider $\beta$ such that $\kappa(\rho,\beta)\leq\frac{1}{4(N-1)}$ then, let $\alpha$ such that $m(2\alpha)\leq 8(N-1)M( \beta/2)$, and suppose $\ep N^2 \ll \alpha^2 m(2\alpha)$. Then, for $\mathbb{P}_{\ep}(x) = | \psi_{\ep}|^2(x)$ we have
\begin{equation}\label{eqn:polysmall2}
 \boxed{\int_{ D_{ \alpha}} \mathbb{P}_\ep(x)\, dx \leq e^{ - \frac 16 \sqrt{\frac{ \alpha^2 m (2 \alpha)}{8\ep}}}. }
 \end{equation}
%Moreover, this asymptotics are uniform in the following sense:   there exist $C$ and $\gamma$, depending only on $\beta$, and not on $\ep$ nor (explicitely) on $\rho$, such that  
%\begin{equation}  \int_{ D_{ \alpha}} |\psi_\ep|^2\, dx \leq \dfrac{C}{\ep^\gamma} 2^{-\ln_2\left(\sqrt{\ep}\right)^2}. \end{equation}
\end{thm}

Notice that in \cite{ColMarStr-19} the diagonal estimate is proven also in the weaker (and sharper) hypotesis $\kappa(\rho, \beta) < \frac 1{N}$: while we believe that also in that case a similar generalization in the case $\ep>0$ holds true, the proof will be more technical and not so trasparent. For the same reason the inequality $\kappa(\rho,\beta)\leq\frac{1}{4(N-1)}$ is used instead of $\kappa(\rho,\beta)<\frac{1}{2(N-1)}$ in order to have more transparent estimates in the end.\\

\begin{rem}[Natural assumption on $\ep$] We want here to justify the hypotesis $\ep N^2 \ll \alpha/16$. Notice that we expect $|\psi_{\ep}|=0$ on the diagonal since $v_{ee}=+\infty$ there; if $|\psi_{\ep}(x)|\sim 1$ on a point in $D_{\alpha}$ then we then expect $|\nabla \psi_{\ep} | \sim \alpha^{-1} $ and so the kinetic energy locally is $ \frac{\ep}{\alpha^2}$ while the potential energy is $\frac 1{\alpha}$. Since the localization phenomenon outside the enlarged diagonal is a classical one, we expect that it holds in the regime where the kinetic energy is negligible with respect to the potential energy, and that is precisely when $\ep \ll \alpha$. Then there is a factor $N$ which is constant, and we believe it is actually technical.
\end{rem}

%{\color{red} Notice that the result is reminiscent of Agmon-type estimates for eigenfunctions of Schr\"{o}dinger operators. We expect in fact that for $\ep \ll 1$ a similar localization result should hold far away from the support of the limit optimal plan; actually an effective potential $E_{eff} \geq 0$ can be defined and $\supp (\mathbb{P}_0) \subseteq \{ E_{eff}=0\}$. We believe that a concentration inequality should hold for $\mathbb{P}_{\ep}  ( \{ E_{eff} > M\} ) $.  }
%
%\todo{L: Esito a metterla come remark ma può essere un po' risky nel senso che il referee ci può dire di dimostrarla...ma si puo anche mettere come oggetto "di un future work"}
%\augusto{Dopo di aver sentito Simone credo che e' meglio toglierla.}

\noindent
\textit{Organization of the paper:} In Section 2 we introduce  the notations we are going to use throughout all the paper. In Section 3 we give some estimates concerning kinetic energy term in the Levy-Lieb functional.
Section 4 is then devoted to the construction of a competitor for the Levy-Lieb functional; finally in Section 5 we derive the diagonal estimates for the wave-function and, thus, prove Theorem \ref{thm:mainThm} and Theorem \ref{thm:mainThm2} via the iteration of a decay estimate.

%
%
%$$F_{LL,\ep}[\rho]:=\min \left\{  \int c \mathbb{P} \, dx  + \ep \Fish(\mathbb{P}) \; : \; \mathbb{P} \in \Pi_N(\rho) \right\} $$
%

\section{Notation}

Consider a subset  $I \subseteq \{ 1, \ldots, N\}$, with cardinality $k=|I|$, defined as  $I=\{ i_1, \ldots, i_k\}$, with $1\leq i_1 < i_2 < \cdots < i_k$. Then, the $I$-projection
$$\pi_I: \R^{3N} \to \R^{3k}, \qquad \pi_I ( (x_1, \ldots , x_N)) = (x_{i_1}, x_{i_2} , \ldots , x_{i_k}).$$
Sometimes we will denote $x_I=\pi_I(x)$ and if $I=J^c$, then $x=(x_I,x_J)$. With a slight abuse of notation, for a function $\mathbb{P} \in L^1(\R^{3N})$,  $I \subseteq \{ 1, \ldots, N\}$  and $J=I^c$ we denote
$$(\pi_I)_{\sharp} ( \mathbb{P}) (x_I)= \int \mathbb{P}(x_I,x_J) \, d x_J,$$
which on density of measures act precisely as the push-forward through the projection function $\pi_I$.\\
As we have  already mentioned above, we denote by $\Pi_N(\rho)$ the set of probability measures on $\R^{3N}$ having the $N$ one body marginals equal to $\rho$.\\
In the following we will consider an electron-electron pair interaction repulsion potential, that $v_{ee}$ with the following form:
\begin{equation}\label{eqn:Vee2} \begin{split}  v_{ee}(x_1, \ldots, x_N)= \sum_{i <j}c(x_i,x_j), \quad \text{ where } m( |x-y|) \leq c(x,y) \leq M(|x-y|) \qquad  \forall x,y \in \R^3  \\ \text{for some } m,M : (0,\infty) \to [0,\infty)\text{ decreasing  such that } \lim_{t \to 0^+} m(t)=+\infty \\    \end{split} \end{equation} 
Moreover, with a slight abuse of notation, we will denote by 
\begin{equation}\label{eqn:Veepot}
\mathbb{P}\in\mathcal P(\R^{3N})\mapsto v_{ee}(\mathbb{P}):=\int_{\R^{3N}}v_{ee}(x_1,\ldots,x_N)\, d\mathbb{P}(x_1,\ldots,x_N)
\end{equation}
Notice that we will often identify a measure $\mathbb{P}$ with its density.

Finally, given an open set $\Omega \subseteq \R^{3N}$, then for every $r>0$ we denote 
\begin{equation}\label{eqn:omegar} \Omega_{-r} := \{ x \in \R^{3N} \; : \; B(x, r) \subseteq \Omega\}.\end{equation}
\section{Estimate for the kinetic energy}
In this section we give some preliminary estimates for the kinetic energy term of the Levy-Lieb functional.\\
Denoting $L_+^1$ the cone of positive $L^1$ functions, we define $\Fish : L_+^1(\R^{3N}) \to \mathbb{R}$ the Kinetic energy associated to some absolutely continuos $N$-probability measure $h$
\begin{equation}\label{eqn:Fish}\Fish(h) :=\begin{cases} \int_{\R^{3N}} \frac { \sum_{i=1}^N | \nabla_i h (x_1, \ldots , x_N) |^2 }{ h (x_1, \ldots, x_N) } \, d x_1, \ldots, d x_N \quad &\text{ if } \sqrt{h} \in H^1( \R^{3N}) \\ +\infty & \text{ otherwise} \end{cases}
\end{equation}
When it will be clear from the context we will also abbreviate $\Fish(h)= \int \frac {| \nabla h|^2}h \,dx $. 
Notice that the kinetic energy functional is also know as the Fisher information.
Moreover if $\psi \in H^1(\R^{3N}; \R)$, then 
\[\int | \nabla \psi|^2 \, d x =\Fish( | \psi|^2)= \Fish ( \mathbb{P}_{\psi}),\]
where $\mathbb{P}_{\psi}=|\psi|^2$ is the joint probability associated to the wave-function $\psi$.
The string of equalities above is thus true when $\psi$ is a minimizer for the bosonic case.
The following Lemma summarises some results concerning the homogeneity, sub-additivity (which is a consequence of theorem 7.8 in \cite{LieLos-01}) and the decomposability under projection of the kinetic energy (a similar result also appears in \cite{Kiessling-12,Rougerie-cdf}). 
\begin{lem}\label{lem:properties} Let $\Fish$ defined as in \eqref{eqn:Fish}. Then
\begin{itemize}
\item[(i)] $\Fish$ is  $1$-homogeneous, that is $\Fish(\lambda \mathbb{P})=\lambda \Fish (\mathbb{P})$ for every $\lambda >0$;
\item[(ii)]  given $\mathbb{P}_1, \ldots, \mathbb{P}_k \in L^1(\R^{3N})$, we have 
$$ \Fish (\mathbb{P}_1+\ldots + \mathbb{P}_k) \leq \Fish(\mathbb{P}_1) + \Fish(\mathbb{P}_2)+ \ldots + \Fish(\mathbb{P}_k);$$
\item[(iii)] Let $\mathbb{P} \in L_+^1(\R^{3N})$. Given $I, J \subseteq \{ 1, \ldots, N\}$ two nonempty disjoint sets such that $I =J^c$, we denote by $\mathbb{P}_I = (\pi_I)_{\sharp} \mathbb{P}$ and $\mathbb{P}_J= (\pi_J)_{\sharp} \mathbb{P}$. Then we have (here $N_I=\sharp I$ and $N_J= \sharp J$)
$$ \Fish^N ( \mathbb{P} ) \geq \Fish^{N_I} ( \mathbb{P}_I) + \Fish^{N_J} ( \mathbb{P}_J), $$
where the equality holds if and only if $\mathbb{P}(x)= \mathbb{P}_I(x_I)\mathbb{P}_J(x_J)/\lambda$, where $\lambda= \int \mathbb{P}$. In particular if $\mathbb{P}$ is the density of a probability measure, we have that the equality happens if and only if $x_I$ and $x_J$ are independent under the probability $\mathbb{P}$.
\end{itemize}
\end{lem}

\begin{proof} $(\rm{i})$ The $1$-homogeneity is obvious.\\
$(\rm{ii})$  For the subadditivity it is sufficient to prove it for $k=2$; then for every $x$, by Cauchy-Schwarz inequality we have
$$ ( \mathbb{P}_1(x) + \mathbb{P}_2(x)) \left( \frac { | \nabla \mathbb{P}_1 (x)|^2}{\mathbb{P}_1(x)}  +  \frac { | \nabla \mathbb{P}_2 (x)|^2}{\mathbb{P}_2(x)} \right) \geq (| \nabla\mathbb{P}_1(x) | + | \nabla \mathbb{P}_2(x)|)^2,$$
which, after using the triangular inequality and dividing by $\mathbb{P}_1+\mathbb{P}_2$ can be rewritten as $$\frac {|\nabla (\mathbb{P}_1+\mathbb{P}_2)|^2}{\mathbb{P}_1+\mathbb{P}_2} \leq \frac { |\nabla \mathbb{P}_1|^2}{\mathbb{P}_1} + \frac { |\nabla \mathbb{P}_2|^2}{\mathbb{P}_2},$$ which integrated gives us the conclusion. 

$(\rm{iii})$ As for the last point we fix $x_J$ and we use the Cauchy-Schwarz inequality with respect to the measure $dx_I$:

\[
\begin{split}
\left( \int \mathbb{P}(x_I, x_J) d x_I \right) \cdot \left( \int \frac { | \nabla_J \mathbb{P}(x_I,x_J) |^2 }{\mathbb{P}(x_I,x_J)} d x_I \right) &\geq \left( \int | \nabla_J \mathbb{P}(x_I,x_J) | d x_I \right)^2\\
& \geq  \left| \nabla_J \Bigl( \int \mathbb{P}(x_I,x_J)  d x_I \Bigr) \right|^2, 
\end{split}
\]
where in the last passage we used the triangular inequality  and we took the derivative out of the integral. Now we recognize $\mathbb{P}_J(x_J)=\int \mathbb{P}(x_I, x_J) d x_I$ and so we can write this as
$$ \int \frac { | \nabla_J \mathbb{P}(x_I,x_J) |^2 }{\mathbb{P}(x_I,x_J)} d x_I  \geq \frac {| \nabla_J \mathbb{P}_J(x_J)|^2}{\mathbb{P}_J(x_J)}. $$
Integrating this with respect to $dx_J$ and doing a similar computation for $x_I$, we obtain the conclusion, that is
$$  \iint \frac { | \nabla \mathbb{P}(x_I,x_J) |^2 }{\mathbb{P}(x_I,x_J)} d x_I  d x_J\geq \int \frac {| \nabla_J \mathbb{P}_J(x_J)|^2}{\mathbb{P}_J(x_J)} d x_J + \int \frac {| \nabla_I \mathbb{P}_I(x_I)|^2}{\mathbb{P}_I(x_I)} d x_I .$$

From the equality cases in C-S  and triangular inequality combined we get $\nabla_J \mathbb{P}(x_I,x_J) = v(x_J) \mathbb{P}(x_I,x_J)$ for some vector field $v$; by a simple integration we actually get $v= \nabla (\mathbb{P}_J) /\mathbb{P}_J$; this can be seen as $\nabla_J \log( \mathbb{P}) = \nabla_J \log \mathbb{P}_J$; similarly we can get $\nabla_I  \log( \mathbb{P}) = \nabla_I \log \mathbb{P}_I $. Summing up this two equalities we get $\nabla (\mathbb{P}(x)/\mathbb{P}_I(x_I)\mathbb{P}_J(x_J)) =0$.

\end{proof}

The following lemma  is a straightforward adaptation of Theorem 3.2 in \cite{CycFroKirSim-87} giving the IMS localization formula; we have added a short proof for sake of completeness.

\begin{lem}\label{lem:energy} Let $\eta_1, \eta_2, \eta_3 : \R^{3N} \to [0,1] $ be $C^1$ functions such that $\eta_1+\eta_2+\eta_3\equiv1$. Then, for every function $\mathbb{P} \in L_+^1(\R^{3N})$ we have
$$ \Fish (\mathbb{P} \eta_1 ) + \Fish (\mathbb{P} \eta_2 )+\Fish (\mathbb{P} \eta_3 ) = \Fish(\mathbb{P}) + \int \left(\frac {| \nabla \eta_1|^2}{\eta_1} + \frac {| \nabla \eta_2|^2}{\eta_2} + \frac {| \nabla \eta_3|^2}{\eta_3}\right) \mathbb{P} \, dx.$$
\end{lem}

\begin{proof} For every $i=1,2,3$ pointwisely we have:

\begin{align*} \frac{| \nabla (\mathbb{P} \eta_i) |^2 }{\mathbb{P} \eta_i}&= \frac{| \eta_i \nabla \mathbb{P} + \mathbb{P} \nabla \eta_i |^2}{\mathbb{P} \eta_i}  = \frac{\eta_i^2 | \nabla \mathbb{P}|^2 + 2 \eta_i \mathbb{P} \nabla \mathbb{P} \cdot \nabla\eta_i+ \mathbb{P}^2 |\nabla \eta_i|^2}{\mathbb{P} \eta_i}  \\
&= \eta_i \frac {|\nabla \mathbb{P}|^2}{\mathbb{P}} + 2 \nabla \mathbb{P} \cdot \nabla \eta_i + \mathbb{P} \frac {|\nabla \eta_i|^2}{\eta_i} 
\end{align*}

Adding them up and using that $\sum \eta_i=1$ and $\sum \nabla \eta_i =0$, we get

$$ \sum_i  \frac{| \nabla (\mathbb{P} \eta_i) |^2 }{\mathbb{P} \eta_i} =  \frac {|\nabla \mathbb{P}|^2}{\mathbb{P}} + \mathbb{P} \sum_i \frac {|\nabla \eta_i|^2}{\eta_i}, $$

which integrated, gives us the desired identity.

\end{proof}

\section{New trial state: swapping particles and estimate for the potential}

The scope of this subsection is to create a competitor for the minimization of the functional 
\begin{equation}\label{eqn:FLLe}
\mathcal{F}_{LL,\ep}(\mathbb{P}) = \begin{cases} \ep \Fish (\mathbb{P}) + v_{ee}(\mathbb{P}) \quad &\text{ if } \mathbb{P} \in \Pi_N(\rho) \\ +\infty &\text{ otherwise,}, \end{cases}
\end{equation}
where $\Fish$ is defined in \eqref{eqn:Fish} and $v_{ee}$ satisfies %\eqref{eqn:Veepot},
 \eqref{eqn:Vee2}. The idea is to try to mimic what it is done in \cite{Pascale-15,ButChaPas-17,ColMarStr-19}, in the semiclassical case $\ep = 0$: in that case we take two points $y,z \in \R^{3N}$ and substitute them with $\tilde{y}, \tilde{z}$ where we have interchanged their first compenent, that is $\tilde{y} = (z_1,y_2, \ldots, y_n)$ and $\tilde{z} =(y_1, z_2, \ldots, z_n)$.

In order to do so for the $n$-particle distribution $\mathbb{P}$, we will consider two small bumps centered around $y$ and $z$
\begin{equation}\label{eqn:etas}
\eta_1(x)= \lambda_1 \eta\Bigl(\frac{x-y}{r_1}\Bigr) \qquad \text{ and }  \qquad \eta_2 (x)= \lambda_2\eta\Bigl(\frac {x-z}{r_2}\Bigr),
\end{equation}
for some $\lambda_1, \lambda_2, r_1, r_2$ to be chosen later and some $\eta \in C^1_c(B(0,1))$, $\eta \geq 0$. First of all we assume that $\supp( \eta_1) \cap \supp (\eta_2) = \emptyset$, which can be granted as long as 
\begin{equation}\label{eqn:rs} r_1+r_2 < |y-z|, \end{equation}
 and then we assume $\int \eta_1\mathbb{P} = \int \eta_2 \mathbb{P} = m$ which can be accomplished again by choosing the appropriate $\lambda_i, r_i$. Let us then define
 \begin{equation}\label{eqn:marginals} \rho_1^i (x_1)= (\pi_{\{1\}})_{\sharp} (\eta_i \mathbb{P}), \qquad \rho_{\hat{1}}^i(x_2,x_3,\ldots, x_N)= (\pi_{\{1\}^c})_{\sharp} ( \eta_i \mathbb{P}), \end{equation}
 \begin{equation}\label{eqn:P1P2} \mathbb{P}_1 = \frac 1m \rho_1^2 \rho_{\hat{1}}^1, \qquad \mathbb{P}_2 = \frac 1m \rho_1^1 \rho_{\hat{1}}^2,\end{equation}
where $\rho_1^i$ and $\rho_{\hat{1}}^i$ are the marginals of $\eta_i\mathbb{P}$ and $\mathbb{P}_1$ and $\mathbb{P}_2$ are densities concentrated around $\tilde{y} = (z_1,y_2, \ldots, y_n)$ and $\tilde{z} =(y_1, z_2, \ldots, z_n)$ respectively. We then finally consider
\begin{equation}\label{eqn:Pbar} \bar{\mathbb{P}}:= \mathbb{P} - \mathbb{P}\eta_1- \mathbb{P} \eta_2 + \mathbb{P}_1 +\mathbb{P}_2, \end{equation}
which will be the competitor for  a minimizer $\mathbb{P}$ of the functional $\mathcal{F}_{LL,\ep}$.

\begin{rem} In order to have that $\bar{\mathbb{P}}$ is a competitor, we still have to check that $\bar{\mathbb{P}} \geq 0$ and that it has the correct marginals. For the positivity, notice that if $\lambda_1$ and $\lambda_2$ are small enough we have also $\eta_1+ \eta_2 \leq 1$ and in particular $\mathbb{P}- \eta_1 \mathbb{P}_1 - \eta_2 \mathbb{P}_2$, which will guarantee that $\bar{\mathbb{P}} \geq 0$.

For the marginal constraint, notice that by \eqref{eqn:marginals} and \eqref{eqn:P1P2} we have that $\eta_1\mathbb{P} + \eta_2 \mathbb{P}$ and $\mathbb{P}_1+\mathbb{P}_2$ have the same marginals, in particular also $\mathbb{P}$ and $\bar{\mathbb{P}}$ share the same marginals.
\end{rem}

\begin{lem}\label{lem:conti} Let $\mathbb{P}$ be such $\mathcal{F}_{LL,\ep}(\mathbb{P})< +\infty$. Given $y,z \in \mathbb{R}^{3N}$, let $\eta_1, \eta_2, \mathbb{P}_1, \mathbb{P}_2, \bar{\mathbb{P}}$ defined by \eqref{eqn:etas},\eqref{eqn:rs}, \eqref{eqn:marginals}, \eqref{eqn:P1P2} and \eqref{eqn:Pbar}. Then

$$\Fish(\bar {\mathbb{P}}) \leq  \Fish( \mathbb{P}) + \int \mathbb{P}(x)\left(\frac{|\nabla \eta_1|^2}{\eta_1} + \frac{| \nabla \eta_1|^2}{\eta_2} +  \frac{|\nabla \eta_1+ \nabla \eta_2|^2}{1-\eta_1- \eta_2}\right) \, dx;$$
$$v_{ee}(\bar{ \mathbb{P}}) =  v_{ee}( \mathbb{P}) - \int \mathbb{P} (\eta_1+\eta_2) \sum_{i>1} c(x_1,x_i) dx  + \int (\mathbb{P}_1+\mathbb{P}_2) \sum_{i>1} c(x_1,x_i) dx. $$

\end{lem}

\begin{proof} Let us consider $\eta_3=1-\eta_1-\eta_2$. Then we have $\bar{\mathbb{P}} = \eta_3\mathbb{P}+ \mathbb{P}_1+\mathbb{P}_2$. Using Lemma \ref{lem:properties}, in particular the subadditivity and the exact energy split in case of independent variables for $\Fish$, we get (by \eqref{eqn:P1P2})

\begin{equation} \label{eqn:Kin1}
\begin{split}
 \Fish(\bar {\mathbb{P}}) &\leq \Fish(\eta_3\mathbb{P}) + \Fish(\mathbb{P}_1) + \Fish(\mathbb{P}_2) \\
 &= \Fish(\eta_3\mathbb{P}) + \Fish(\rho_1^2) + \Fish (\rho_{\hat{1}}^1) + \Fish(\rho_1^1) + \Fish (\rho_{\hat{1}}^2);
 \end{split}
\end{equation}
we then recall \eqref{eqn:marginals} and the inequality for the split energy (Lemma \ref{lem:properties} (iii)) to get 
\begin{equation}\label{eqn:Kin2} 
\Fish(\rho_1^i) +\Fish(\rho_{\hat{1}}^i) \leq \Fish (\eta_i \mathbb{P})\end{equation}
 and so we conclude using \eqref{eqn:Kin1}, \eqref{eqn:Kin2} and then Lemma \ref{lem:energy}.
 
 For the estimate with the potential, it is clear that
$$v_{ee}(\bar{ \mathbb{P}}) =  v_{ee}( \mathbb{P}) - \int \mathbb{P} (\eta_1+\eta_2) v_{ee}(x) dx  + \int (\mathbb{P}_1+\mathbb{P}_2) v_{ee}(x) dx; $$
Since $v_{ee}(x)= \sum_{i<j} c(x_i,x_j)$ we just need to show that the contribution due to $c(x_i,x_j)$ whenever $1<i<j$ cancels out in the last two integrals. In fact in both integrals we can integrate out the first variable: denoting $I=\{1\} $ and $J=I^c$ for example we have

\begin{align*}
\int \mathbb{P} \eta_1 c(x_i,x_j) \, d x_I dx_ J &= \int c(x_i,x_j) \left( \int \mathbb{P} \eta_1 \, d x_I  \right) \, d x_J\\ &= \int c(x_i,x_j) \rho_{\hat{1}}^1 (x_J)  \, d x_J \\
&= \int c(x_i,x_j) \rho_{\hat{1}}^1(x_J)\left(\int \frac{ \rho_{2}^1(x_I)}m \, d x_I\right) \, d x_J \\&= \int  c(x_i,x_j)  \mathbb{P}_1 \, dx.
\end{align*}
In a similar way we can show that $\int \mathbb{P} \eta_2 c(x_i,x_j) \,dx = \int \mathbb{P}_2 c(x_i,x_j) \, dx$.
\end{proof}

In the sequel we will denote $C_1(x)= \sum_{i=2}^N c(x_1,x_i)$

%it will be important to consider this: if $y \in D_{\alpha}$ and $\tilde{y} \not \in D_{\beta}$ then
% $$ \int C \mathbb{P}  \eta_1 dx \geq m(\alpha+2r_1) \int \mathbb{P} \eta_1  \qquad , \qquad  \int C f_1 dx \leq M(\beta-2r_1) \int \mathbb{P} \eta_1 $$
\section{Diagonal estimates for the wave-function}
We devote this last section to derive the diagonal estimates for the bosonic wave-function which minimizes the Levy-Lieb functional proving in particular Theorem \ref{thm:mainThm} and Theorem \ref{thm:mainThm2}.

\begin{lem}\label{lem:exdoubling}  Let $\rho$ be a marginal distribution and let $\beta>0$  be such that $\kappa(\rho, \beta) \leq \frac 1{4(N-1)}$. Then,  for every $\mathbb{P} \in \Pi_N(\rho)$ and $y \in \R^{3N}$, for every $r_1,r_2$ such that  $r_1+ 2r_2 < \beta$ and $\delta>0$, there exists $z \in \R^{3N}$ such that, defining $ \eta_1, \eta_2, \mathbb{P}_1, \mathbb{P}_2, m$ as in \eqref{eqn:etas},  \eqref{eqn:marginals} and \eqref{eqn:P1P2}\begin{itemize}
\item[(i)] $\int C_1 (\mathbb{P}_1+\mathbb{P}_2) dx \leq 2(N-1) M(\beta-r_1-2r_2) m $;
\item[(ii)] $z$ is a $(1+\delta,1/2)$-doubling point at scale $r_2$ for $\mathbb{P}$, that is \[\int_{B(z,r_2)} \mathbb{P} \, dx \leq 2 (1+\delta)^{3N} \int_{B(z,r_2/(1+\delta))} \mathbb{P} \, dx.\]
\end{itemize}
% $$ \int C_1 \mathbb{P}  \eta_1 dx \geq m(|y_1-y_2|+2r_1) m   \qquad  \int C_1 (f_1+f_2) dx \leq 4(N-1) M(\beta-2\max\{r_1,r_2\}) m $$
\end{lem}

\begin{proof} 
For $\gamma>0$, let us consider the set 
$$\Omega = \left\{ y' \in \R^{3N} : | y'_1 - y_i| \geq \gamma-\frac{\delta}{1+\delta} r_2 \text{ and } |y_1 - y'_i | \geq \gamma- \dfrac{\delta}{1+\delta} r_2 , \forall i=2, \ldots, N \right\}.$$
We know that if $z \in \Omega$ we will have of course 
$$C_1(y'_1, z'_2, \ldots, z_N')+C_1(z'_1, y'_2, \ldots, y'_N) \leq 2 (N-1) M(\gamma - r_1-2r_2) \;\forall y' \in B(y,r_1),  z' \in B(z,r_2),$$
which in particular implies $\int C_1(\mathbb{P}_1+\mathbb{P}_2) \, dx \leq 2 (N-1) M(\gamma - r_1 -2r_2) m$. Now we want to see that there exists a $1/2$ doubling point in $\Omega$; in order to do that, it is easy to see that 
$$ \Omega_{- \delta \frac{ r_2}{1+\delta}} \subseteq \{ y' \in \R^{3N} : | y'_1 - y_i| \geq \gamma \text{ and } |y_1 - y'_i | \geq \gamma , \forall i=2, \ldots, N \} $$
And now a similar computation to what is done in  \cite{ButChaPas-17,ColMarStr-19} will give us
$$ \int_{\Omega_{-\delta\frac{r_2}{1+\delta}}} \mathbb{P}(x) \, dx \geq 1- 2(N-1)\kappa (\rho , \gamma), $$
where $\kappa (\rho, r) = \sup_{ x \in \R^3} \rho( B(x,r))$. Now if we consider $\gamma=\beta$ we have $\kappa(\rho, \beta) \leq \frac 1{4(N-1)}$, and so we can apply Lemma \ref{lem:doubling} with $r=\frac{r_2}{1+\delta}$ get the existence of a $(1+\delta,1/2)$-doubling point at scale $r_2$ in $\Omega$.

\end{proof}

\begin{lem}[Existence of doubling points] \label{lem:doubling} Let $\mathbb{P} \in L^1_+( \R^{3N})$ be the density of a probability measure and let $\delta>0$. Let us consider an open set $\Omega \subseteq \R^{3N}$; then for every $r>0$ we denote  $C_r :=  \int_{\Omega_{-r}} \mathbb{P}(x) \, dx$, where $\Omega_{-r}$ is defined as in \eqref{eqn:omegar}.
Then, whenever $C_{\delta r}>0$, there exists $y \in \Omega$, such that 
$$ \int_{B(y,(1+\delta) r)} \mathbb{P}(x)\, dx \leq \frac {(1+\delta)^{3N}}{C_{\delta r}}  \int_{B(y,r)} \mathbb{P}(x) \, dx,$$
that is, the measure $\mathbb{P}(x) \, dx$ is doubling at the point $y$ at scale $r$, with doubling constant $\frac {(1+\delta)^{3N}}{C_{\delta r}} $.
\end{lem}
\begin{proof} Suppose on the contrary that for every $y \in \Omega$ the reversed inequality holds
$$ \int_{B(y,(1+\delta)r)} \mathbb{P}(x)\, dx > \frac {(1+\delta)^{3N}}{C_{\delta r}}  \int_{B(y,r)} \mathbb{P}(x) \, dx.$$
Then we can integrate this inequality on the whole $\Omega$
$$   \int_{\Omega} \int_{B(y,(1+\delta)r)} \mathbb{P}(x)\, dx  \, dy> \frac {(1+\delta)^{3N}}{C_{\delta r}}\int_{\Omega} \int_{B(y,r)} \mathbb{P}(x) \, dx \, dy.$$
Now we can use Fubini and get
\begin{align*}
\omega_{3N} \cdot \bigl((1+\delta)r\bigr)^{3N}  &=  \int_{\R^{3N}} \int_{B(y,(1+\delta)r)} \mathbb{P}(x)\, dx  \, dy \geq  \int_{\Omega} \int_{B(y,(1+\delta)r)} \mathbb{P}(x)\, dx  \, dy \\
C_{\delta r} \omega_{3N} \cdot r^{3N} & = \int_{ \Omega_{-\delta r} } \mathbb{P}(z) | B(z,r)| \, d z = \int_{\Omega} \int_{B(y,r) \cap \Omega_{-\delta r}} \mathbb{P}(x) \, dx \, dy \\
& \leq   \int_{\Omega} \int_{B(y,r)} \mathbb{P}(x) \, dx \, dy,
\end{align*}
and so we get a contradiction.
\end{proof}

\begin{prop}[One step decay]\label{prop:decay} Let us consider $\rho$ and $\beta$ such that $\kappa(\rho, \beta) < \frac 1{4(N-1)}$. Then there exists $\alpha_0= \alpha(\beta, \ep)$ such that if $\mathbb{P} $ minimizes \eqref{eqn:FLLe}, we have that for every $y \in D_{\alpha}$ such that $\alpha \leq \alpha_0$, and every $r_1 \leq \alpha/2$, we have
\begin{equation}\label{eqn:decay1} \int_{ B(y, r_1/(1+\delta)) } \mathbb{P}(x)\, dx \leq \frac 1{\delta^2 r_1^2 \frac{ m(2\alpha)}{2(1+\delta)^2 \ep} +1 } \int_{ B(y, r_1) } \mathbb{P}(x) \, dx,\end{equation}
whenever $\delta>0$ is such that $m(2\alpha) >256\ep C(\delta)/\beta^2$, where 
\begin{equation}\label{eqn:Cdelta} C(\delta):=\frac {(1+\delta)^2 \cdot  (2(1+\delta)^{3N}-1)}{\delta^2}.\end{equation}
An implicit choice for $\alpha_0$ is for example $ m (2\alpha_0) > 8\max\{ (N-1)M(\beta/2) ,  850 \ep N^2 / \beta^2\}$.
\end{prop}

\begin{proof} Let $y \in D_{\alpha}$ and without loss of generality we can assume that $|y_1-y_2|<\alpha$; let $z$ given by Lemma \ref{lem:exdoubling}. We then consider $r_1, r_2, \eta_1, \eta_2, \mathbb{P}_1, \mathbb{P}_2, \bar{\mathbb{P}}$ defined by \eqref{eqn:etas},\eqref{eqn:rs}, \eqref{eqn:marginals}, \eqref{eqn:P1P2} and \eqref{eqn:Pbar}; being $\bar{\mathbb{P}} \in \Pi(\rho)$, we get, by the minimality of $\mathbb{P}$,
$$ \mathcal{F}_{LL,\ep} (\bar{\mathbb{P}} ) \geq  \mathcal{F}_{LL,\ep} (\mathbb{P}), $$
$$ \ep \Fish (\bar{\mathbb{P}}) + v_{ee}(\bar{\mathbb{P}}) \geq  \ep \Fish (\mathbb{P}) + v_{ee}(\mathbb{P}); $$
now we can use the estimates in Lemma \ref{lem:conti} in order to conclude that 
$$\ep\int \mathbb{P} \left(  \frac{| \nabla \eta_1|^2}{\eta_1} +  \frac{| \nabla \eta_2|^2}{\eta_2} +  \frac{|\nabla \eta_1+ \nabla \eta_2|^2}{1-\eta_1- \eta_2}\right) \, dx \geq\int \mathbb{P} (\eta_1+\eta_2) C_1 dx - \int (\mathbb{P}_1+\mathbb{P}_2) C_1 dx .$$

Now we make the choice $\eta(x)= \min\left\{\frac{(1+\delta)(1-|x|)_+}{\delta}, 1\right\} ^2$. In particular $0 \leq \eta \leq  1$ and $\eta \equiv 1$ if $|x|<\frac 1{1+\delta}$, and moreover $|\nabla \sqrt{\eta}| \equiv \frac{1+\delta}\delta  $ if $\frac 1{1+\delta}\leq |x| \leq 1$ and $0$ otherwise. We thus have

\begin{align*}
 \int \mathbb{P} \frac{| \nabla \eta_1|^2}{\eta_1} &= \frac {(1+\delta)^2}{\delta^2r_1^2} \int_{ B (y, r_1) \setminus B(y, r_1/(1+\delta))} \mathbb{P} \lambda_1 \, dx  =  \frac {(1+\delta)^2}{\delta^2r_1^2} \left(  \int_{B(y, r_1)} \mathbb{P} \lambda_1\, dx - \int_{B(y,r_1/(1+\delta))} \mathbb{P} \lambda_1 \right)\\
 & \leq  \frac {(1+\delta)^2}{\delta^2r_1^2} \left(  \int_{B(y, r_1)} \mathbb{P} \lambda_1\, dx - m \right)
 \end{align*}
 %\todo{L : wondering if there should be a 16 instead of 4}
 In a similar way we have
 
 $$ \int \mathbb{P} \frac{| \nabla \eta_2|^2}{\eta_2} \leq \frac {(1+\delta)^2}{\delta^2 r_2^2} \left(  \int_{B(z, r_2)} \mathbb{P} \lambda_2\, dx - m \right) \leq   \frac {(1+\delta)^2 \cdot  (2(1+\delta)^{3N}-1)}{\delta^2r_2^2} \cdot m = \frac{C(\delta)}{r_2^2} \cdot m.$$
where in the last steps we used Lemma \ref{lem:exdoubling} (ii) and the definition of $C(\delta)$ \eqref{eqn:Cdelta}.  Now we use that $\int C_1 \mathbb{P} \eta_1  \, dx \geq m(\alpha + 2r_1) \cdot m$ and the estimates we have for $\int C_1(\mathbb{P}_1+\mathbb{P}_2) \, dx  $ to get 
$$\int \mathbb{P} (\eta_1+\eta_2) C_1 dx - \int (\mathbb{P}_1+\mathbb{P}_2) C_1 dx \geq [m (\alpha +2r_1)   - 2(N-1)M(\beta -r_1 - 2r_2)] \cdot m.$$
Putting everything together we have
$$ \frac {\ep(1+\delta)^2}{\delta^2r_1^2} \int_{B(y,r_1)} \lambda_1 \mathbb{P}(x)\, dx \geq \left[ m (\alpha +2r_1)   - 2(N-1)M(\beta -r_1 - 2r_2) - \frac {\ep C(\delta)}{r_2^2} + \frac {\ep(1+\delta)^2}{\delta^2r_1^2} \right] \cdot m.$$
We can now use use $m \geq \int_{B(y,r_1/(1+\delta))} \lambda_1 \mathbb{P}(x) \, dx$ and, dividing $\lambda_1$, we can write the inequality as
\begin{equation}\label{eqn:F} \int_{B(y,r_1/(1+\delta))}   \mathbb{P} (x) \, dx \leq  \frac { 1}{ \frac{\delta^2 r_1^2 F(r_1,\ep, \alpha) }{(1+\delta)^2\ep} + 1 }\int_{B(y,r_1)}   \mathbb{P}(x)\, dx,\end{equation}
where 
\[F(r_1, \ep, \alpha) := \max \left\{ m (\alpha +2r_1)   - 2(N-1)M(\beta -r_1 - 2r_2) - \frac {\ep C(\delta)}{r_2^2} \; : \; r_2>0\right\}.\]
 Now we can choose $r_1 \leq \alpha/2 \leq \beta/4$ and $r_2= \beta/8$, and then choose $\alpha< \alpha_0$ such that 
\begin{equation}\label{eqn:conditionsalpha} \frac{m (2\alpha)}{2}   - 2(N-1)M(\beta/2) > \frac{m(2 \alpha)}4 \quad \text{ and }\quad  \frac{m (2\alpha)}{2}   - \frac {\ep C(\delta)}{r_2^2} > \frac{m(2 \alpha)} 4.\end{equation}
In this way we have $F(r_1, \ep, \alpha) \geq m(2 \alpha)/2$: plugging this estimate in \eqref{eqn:F} we get precisely
$$ \int_{B(y,r_1/(1+\delta))} \mathbb{P} (x) \, dx \leq  \frac { 1}{ \delta^2 r_1^2 \frac{m(2\alpha)}{2(1+\delta)^2\ep} + 1 }\int_{B(y,r_1)}  \mathbb{P}(x)\, dx.$$
In order to understand for which $\alpha$ and $\delta$ this inequality holds, we have to ensure that the two conditions \eqref{eqn:conditionsalpha} are satisfied, that is
\begin{equation}\label{eqn:newcondalpha}
m(2\alpha) \geq \max \left\{ 8 (N-1) M\Bigl( \frac{\beta}2\Bigr),  256 \frac{\ep C(\delta)}{\beta^2} \right\};
\end{equation}
notice that $\alpha_0$ can be characterized as maximal $\alpha$ for which there exists some $\delta$ for which \eqref{eqn:newcondalpha} is satisfied that is when $C(\delta)$ as small as possible, which is approximately achieved for $\delta=\frac 2{3N}$. With this choice we have $C(2/(3N)) \leq 26 N^2$ and thus
\begin{equation}\label{eqn:newcondalpha2}
m(2\alpha_0) \geq 8 \max \left\{  (N-1) M\Bigl( \frac{\beta}2\Bigr),  850 \frac{\ep N^2}{\beta^2} \right\}.
\end{equation}

\end{proof}

We will now iterate the estimate in Proposition \ref{prop:decay}

\begin{thm}\label{thm:estimates} Let us consider $\rho$ and $\beta$ such that $\kappa(\rho, \beta) < \frac 1{4(N-1)}$. Then let us consider $\alpha < \alpha_0$ (as in Proposition \ref{prop:decay}) and suppose $A:= \frac { \alpha^2 m(2\alpha)}{8\ep} \gg N^2$. Then if $\mathbb{P}$ minimizes \eqref{eqn:FLLe} we have that 
%$$  \int_{ B(z, \sqrt{ \ep /m(2\alpha)} ) } \mathbb{P}(x)\, dx \leq C \cdot 2^{-\ln_2\left(\alpha\sqrt{ 8 m(2\alpha)/\ep}\right)^2 } \int_{ B(z, \alpha) } \mathbb{P}(x) \, dx. $$
%$$  \int_{ B(z, a ) } \mathbb{P}(x)\, dx \leq C \cdot 2^{-\ln_2\left(b/a\right) \cdot \ln_2 ( ab / a_0^2) } \int_{ B(z, b) } \mathbb{P}(x) \, dx. $$
%$$ \int_{ D_{ \alpha/2}} \mathbb{P}(x)\, dx \leq \frac {2^{3N}}{\alpha^2 \frac{ m(2\alpha)}{4 \ep} +1 }  \int_{ D_{2\alpha} } \mathbb{P}(x) \, dx. $$
$$ \int_{ D_{ \alpha/2}} \mathbb{P}(x)\, dx \leq e^{-\frac 16\sqrt{\frac { \alpha^2 m(2\alpha)}{8\ep}}}  \int_{ D_{2\alpha} } \mathbb{P}(x) \, dx. $$
%In particular for every $\alpha >0$ that satisfies  .... for every $k>0$ we have $\int_{ D_{ \alpha/2}} \mathbb{P}_{\ep}(x) dx = o (  \ep^k ) $;  more precisely we could also take $k=k(\ep) = -\eta \ln_2(\ep)$ for every $0<\eta<1$.
\end{thm}

\begin{proof}

Let us consider $\delta$ such that $\delta^2A = e^2$. By the hypothesis on $A$ we have $\delta \ll 1/N$; in particular, by \eqref{eqn:Cdelta} we can estimate $C(\delta) \leq \frac 2{\delta^2}$, and then it is easy to see that $m(2\alpha) >256\ep C(\delta)/\beta^2$ and thus we can apply Proposition \ref{prop:decay} with $r_1=\alpha_k=\frac  \alpha2 (1+\delta)^{-k}$ to obtain for every $y \in D_{\alpha}$

\begin{equation} \int_{ B(y,\alpha_{k+1}) } \mathbb{P}(x)\, dx \leq \frac 1{\delta^2 \alpha^2 \frac{ m(2\alpha)}{8(1+\delta)^{2k+2} \ep} +1 } \int_{ B(z, \alpha_k) } \mathbb{P}(x) \, dx \leq  \frac{ (1+\delta)^{2k+2}}{e^2} \int_{ B(y, \alpha_k) } \mathbb{P}(x) \, dx.\end{equation}
We can now iterate the estimate for $k=0, \ldots, k_0$ where $(1+\delta)^{2k_0+2} \leq e^2 \leq (1+\delta)^{2k_0+4}$. At that point we have
\begin{align*} \int_{ B(y,\alpha/2e) } \mathbb{P}(x)\, dx &\leq \int_{ B(y,\alpha_{k_0+1}) } \mathbb{P}(x)\, dx \leq \frac{ (1+\delta)^{(k_0+1)(k_0+2)}}{(e^2)^{k_0+1}} \int_{ B(y, \alpha_0) } \mathbb{P}(x) \, dx \\& \leq e^{-k_0} \int_{ B(y, \alpha/2) } \mathbb{P}(x) \, dx.\end{align*}
Integrating this inequality for $y \in D_{\alpha}$ we get
\begin{align*}
\omega_{3N} \left( \frac {\alpha}{2e} \right)^{3N} \int_{D_{\alpha/2}} \mathbb{P} (y) \, dy & \leq \int_{D_{\alpha}} \int_{ B(y,\alpha/2e) } \mathbb{P}(x)\, dx \, dy \\  &\leq e^{-k_0}  \int_{D_{\alpha}} \int_{ B(y,\alpha/2) } \mathbb{P}(x)\, dx \, dy  \\ & \leq e^{-k_0} \omega_{3N}  \left( 2\alpha  \right)^{3N}\int_{D_{2\alpha}} \mathbb{P}(y) \, dy.
\end{align*}
Now we notice that $k_0+2 \geq \frac{\ln ( e^2)} {2\ln(1+\delta)} \geq \frac {2}{4\delta} = \frac {\sqrt{A}}{2e}$ and so $e^{-k_0} \leq 10 e^{- \frac{\sqrt{A}}{2e}}$. In particular
$$ \int_{D_{\alpha/2}} \, \mathbb{P} (y) \, dy \leq 10e^{ - \frac{\sqrt{A}}{2e} + 3N\ln(4e) }  \int_{D_{2\alpha}} \, \mathbb{P} (y) \, dy;$$
notice that since $A \gg N^2$ we have $\frac{\sqrt{A}}{2e} - 3N\ln(4e) \geq \frac{ \sqrt{A}}6$.
\end{proof}

\begin{proof}(Theorem \ref{thm:mainThm} and Theorem \ref{thm:mainThm2}) First we notice that if $\psi_{\ep}$ is a minimizer for \eqref{Fll} then $\mathbb{P}_{\ep}  = |\psi_{\ep}|^2$ is a minimizer for \eqref{eqn:FLLe}. Then we notice that if $m(2\alpha)\leq 8(N-1)M( \beta/2)$ and  $\ep N^2 \ll \alpha^2 m(2\alpha)$, we have also $\alpha < \alpha_0$ and so we can apply Theorem \ref{thm:estimates}. From that we finish using that $\mathbb{P}_{\ep}$ is a probability density and so $\int_{D_{2\alpha}} \mathbb{P}_{\ep}(y) \, dy \leq 1$. The conclusions for Theorem \ref{thm:mainThm} are then implied by using $m(t)=M(t)=1/t$.
\end{proof}

\noindent{\textbf{Acknowledgement.}} Part of this work has been developed during a Research in Pairs program where the authors were hosted at the MFO (Mathematisches Forschungsinstitut Oberwolfach) in January 2017.  S.D.M. is a member of GNAMPA (INdAM). A.G. acknowledges partial support of his research by the Canada Research Chairs Program and Natural Sciences and Engineering Research Council of Canada. L.N. is partially on academic leave at Inria (team Matherials) for the year 2022-2023 and acknowledges the hospitality if this institution during this period. His work was supported by a public grant as part of the Investissement d'avenir project, reference ANR-11-LABX-0056-LMH, LabEx LMH and  from H-Code, Université Paris-Saclay. The authors want to thank M. Lewin for comments (and references) on a preliminary draft of the paper.

%%%%%%
%\section*{Acknowledgements}

%%%%%%

\bibliographystyle{plain}
\bibliography{biblio}

\begin{thebibliography}{10}

\bibitem{BenCarNen-SMCISE-16}
Jean-David Benamou, Guillaume Carlier, and Luca Nenna.
\newblock A numerical method to solve multi-marginal optimal transport problems
  with coulomb cost.
\newblock pages 577--601, 2016.

\bibitem{BinPas-17}
Ugo Bindini and Luigi De~Pascale.
\newblock Optimal transport with {C}oulomb cost and the semiclassical limit of
  density functional theory.
\newblock {\em J. \'{E}c. polytech. Math.}, 4:909--934, 2017.

\bibitem{BinDepKau-arxiv-20}
Ugo Bindini, Luigi De~Pascale, and Anna Kausamo.
\newblock On {S}eidl-type maps for multi-marginal optimal transport with
  {C}oulomb cost.
\newblock {\em arXiv preprint arXiv:2011.05063}, 2020.

\bibitem{ButChaPas-17}
G.~{Buttazzo}, T.~{Champion}, and L.~{De Pascale}.
\newblock Continuity and estimates for multimarginal optimal transportation
  problems with singular costs.
\newblock {\em Appl. Math. Optim.}, August 2017.

\bibitem{ButPasGor-12}
Giuseppe Buttazzo, Luigi {De Pascale}, and Paola Gori-Giorgi.
\newblock Optimal-transport formulation of electronic density-functional
  theory.
\newblock {\em Phys. Rev. A}, 85:062502, Jun 2012.

\bibitem{CheFri-MMS-15}
Huajie Chen and Gero Friesecke.
\newblock Pair densities in density functional theory.
\newblock {\em Multiscale Modeling \& Simulation}, 13(4):1259--1289, 2015.

\bibitem{ColPasMar-15}
Maria Colombo, Luigi {De Pascale}, and Simone {Di Marino}.
\newblock Multimarginal optimal transport maps for one-dimensional repulsive
  costs.
\newblock {\em Canad. J. Math.}, 67:350--368, 2015.

\bibitem{ColMar-15}
Maria Colombo and Simone {Di Marino}.
\newblock Equality between {M}onge and {K}antorovich multimarginal problems
  with {C}oulomb cost.
\newblock {\em Ann. Mat. Pura Appl. (4)}, 194(2):307--320, 2015.

\bibitem{ColMarStr-19}
Maria Colombo, Simone Di~Marino, and Federico Stra.
\newblock Continuity of multimarginal optimal transport with repulsive cost.
\newblock {\em SIAM J. Math. Anal.}, 51(4):2903--2926, 2019.

\bibitem{ColMarStr-21_ppt}
Maria {Colombo}, Simone {Di Marino}, and Federico {Stra}.
\newblock {First order expansion in the semiclassical limit of the Levy-Lieb
  functional}.
\newblock {\em arXiv e-prints}, page arXiv:2106.06282, June 2021.

\bibitem{ColStr-M3AS-15}
Maria Colombo and Federico Stra.
\newblock Counterexamples in multimarginal optimal transport with coulomb cost
  and spherically symmetric data.
\newblock {\em Mathematical Models and Methods in Applied Sciences},
  26(06):1025--1049, 2016.

\bibitem{CotFriKlu-13}
Codina Cotar, Gero Friesecke, and Claudia Kl{\"u}ppelberg.
\newblock Density functional theory and optimal transportation with {C}oulomb
  cost.
\newblock {\em Comm. Pure Appl. Math.}, 66(4):548--599, 2013.

\bibitem{CotFriKlu-18}
Codina Cotar, Gero Friesecke, and Claudia Kl\"{u}ppelberg.
\newblock Smoothing of transport plans with fixed marginals and rigorous
  semiclassical limit of the {H}ohenberg-{K}ohn functional.
\newblock {\em Arch. Ration. Mech. Anal.}, 228(3):891--922, 2018.

\bibitem{CoyEhrLom-19}
Rafael Coyaud, Virginie Ehrlacher, Damiano Lombardi, et~al.
\newblock Approximation of optimal transport problems with marginal moments
  constraints.
\newblock Technical report, 2019.

\bibitem{CycFroKirSim-87}
H.~L. Cycon, R.~G. Froese, W.~Kirsch, and B.~Simon.
\newblock {\em Schr{\"o}dinger operators with application to quantum mechanics
  and global geometry}.
\newblock Texts and Monographs in Physics. Springer-Verlag, Berlin, study
  edition, 1987.

\bibitem{Pascale-15}
Luigi De~Pascale.
\newblock Optimal transport with {C}oulomb cost. {A}pproximation and duality.
\newblock {\em ESAIM Math. Model. Numer. Anal.}, 49(6):1643--1657, 2015.

\bibitem{MarGerNen-17}
Simone {Di Marino}, Augusto Gerolin, and Luca Nenna.
\newblock {\em Optimal Transportation Theory with Repulsive Costs}, volume
  ``Topological Optimization and Optimal Transport in the Applied Sciences'' of
  {\em Radon Series on Computational and Applied Mathematics}, chapter~9, pages
  204--256.
\newblock De Gruyter, June 2017.

\bibitem{DiMGerNenSeiGor-xxx-15}
Simone Di~Marino, Augusto Gerolin, Luca Nenna, Michael Seidl, and Paola
  Gori-Giorgi.
\newblock in preparation.

\bibitem{MarLewNen-19_ppt}
Simone Di~Marino, Mathieu Lewin, and Luca Nenna.
\newblock Grand-canonical optimal transport.
\newblock {\em arXiv preprint arXiv:2201.06859}, 2022.

\bibitem{Friesecke-19}
Gero Friesecke.
\newblock A simple counterexample to the {M}onge ansatz in multimarginal
  optimal transport, convex geometry of the set of {K}antorovich plans, and the
  {F}renkel--{K}ontorova model.
\newblock {\em SIAM Journal on Mathematical Analysis}, 51(6):4332--4355, 2019.

\bibitem{FriGerGor-arxiv-22}
Gero Friesecke, Augusto Gerolin, and Paola Gori-Giorgi.
\newblock The strong-interaction limit of density functional theory.
\newblock {\em arXiv preprint arXiv:2202.09760}, 2022.

\bibitem{FriSchVoe-21}
Gero Friesecke, Andreas~S. Schulz, and Daniela V\"ogler.
\newblock Genetic column generation: Fast computation of high-dimensional
  multi-marginal optimal transport problems.
\newblock {\em to appear in SIAM J. Sci. Comp., arXiv preprint:
  arXiv:2103.12624}, 2021.

\bibitem{GerGroGor-JCTC-20}
Augusto Gerolin, Juri Grossi, and Paola Gori-Giorgi.
\newblock Kinetic correlation functionals from the entropic regularisation of
  the strictly-correlated electrons problem.
\newblock {\em Journal of Chemical Theory and Computation}, 16(1):488--498,
  2019.

\bibitem{GerKauRaj-ESAIMCOCV-19}
Augusto Gerolin, Anna Kausamo, and Tapio Rajala.
\newblock Duality theory for multi-marginal optimal transport with repulsive
  costs in metric spaces.
\newblock {\em ESAIM: Control, Optimisation and Calculus of Variations}, 25:62,
  2019.

\bibitem{KhoLinLinLex-19_ppt}
Yuehaw Khoo, Lin Lin, Michael Lindsey, and Lexing Ying.
\newblock Semidefinite relaxation of multi-marginal optimal transport for
  strictly correlated electrons in second quantization, 2019.

\bibitem{Kiessling-12}
Michael K.-H. Kiessling.
\newblock The {H}artree limit of {B}orn's ensemble for the ground state of a
  bosonic atom or ion.
\newblock {\em J. Math. Phys.}, 53(9):095223, 2012.

\bibitem{LanMarGerLeeGor-PCCP-16}
G.~Lani, S.~Di~Marino, A.~Gerolin, R.~van Leeuwen, and P.~Gori-Giorgi.
\newblock The adiabatic strictly-correlated-electrons functional: kernel and
  exact properties.
\newblock {\em Phys. Chem. Chem. Phys.}, 18:21092--21101, 2016.

\bibitem{Levy-79}
Mel Levy.
\newblock Universal variational functionals of electron densities, first-order
  density matrices, and natural spin-orbitals and solution of the
  $v$-representability problem.
\newblock {\em Proc. Natl. Acad. Sci. U. S. A.}, 76(12):6062--6065, 1979.

\bibitem{Lewin-18}
Mathieu Lewin.
\newblock Semi-classical limit of the {L}evy-{L}ieb functional in {D}ensity
  {F}unctional {T}heory.
\newblock {\em C. R. Math. Acad. Sci. Paris}, 356(4):449--455, 2018.

\bibitem{lewin2019universal}
Mathieu Lewin, Elliott~H Lieb, and Robert Seiringer.
\newblock Universal functionals in density functional theory.
\newblock {\em arXiv preprint arXiv:1912.10424}, 2019.

\bibitem{Lieb-83b}
Elliott~H. Lieb.
\newblock Density functionals for {C}oulomb systems.
\newblock {\em Int. J. Quantum Chem.}, 24:243--277, 1983.

\bibitem{LieLos-01}
Elliott~H. Lieb and Michael Loss.
\newblock {\em Analysis}, volume~14 of {\em Graduate Studies in Mathematics}.
\newblock American Mathematical Society, Providence, RI, 2nd edition, 2001.

\bibitem{MarGer-19_ppt}
Simone~Di Marino and Augusto Gerolin.
\newblock An optimal transport approach for the schrödinger bridge problem and
  convergence of sinkhorn algorithm, 2019.

\bibitem{MirUmrMorGor-JCP-14}
A.~Mirtschink, C.~J. Umrigar, J.~D. Morgan~III, and P.~Gori-Giorgi.
\newblock Energy density functionals from the strong-coupling limit applied to
  the anions of the he isoelectronic series.
\newblock {\em J. Chem. Phys.}, 140(18):18A532, 2014.

\bibitem{Rougerie-cdf}
Nicolas Rougerie.
\newblock {\em Th{\'e}or{\`e}mes de de {F}inetti, limites de champ moyen et
  condensation de {B}ose-{E}instein}.
\newblock Spartacus-idh, Paris, 2016.
\newblock Cours Peccot au Coll{\`e}ge de France (2014).

\bibitem{SeiMarGerNenGieGor-17}
M.~{Seidl}, S.~{Di Marino}, A.~{Gerolin}, L.~{Nenna}, K.~J.~H. {Giesbertz}, and
  P.~{Gori-Giorgi}.
\newblock The strictly-correlated electron functional for spherically symmetric
  systems revisited.
\newblock {\em ArXiv e-prints}, February 2017.

\bibitem{Seidl-99}
Michael Seidl.
\newblock Strong-interaction limit of density-functional theory.
\newblock {\em Phys. Rev. A}, 60:4387--4395, Dec 1999.

\bibitem{SeiGorSav-07}
Michael Seidl, Paola Gori-Giorgi, and Andreas Savin.
\newblock Strictly correlated electrons in density-functional theory: A general
  formulation with applications to spherical densities.
\newblock {\em Phys. Rev. A}, 75:042511, Apr 2007.

\bibitem{SeiPerLev-99}
Michael Seidl, John~P. Perdew, and Mel Levy.
\newblock Strictly correlated electrons in density-functional theory.
\newblock {\em Phys. Rev. A}, 59:51--54, Jan 1999.

\bibitem{VucGer-2022}
Stefan Vuckovic, Augusto Gerolin, Timothy~J Daas, Hilke Bahmann, Gero
  Friesecke, and Paola Gori-Giorgi.
\newblock Density functionals based on the mathematical structure of the
  strong-interaction limit of dft.
\newblock {\em Wiley Interdisciplinary Reviews: Computational Molecular
  Science}, page e1634, 2022.

\end{thebibliography}

%\begin{thebibliography}{99}
%

%\end{thebibliography}
\end{document}